\documentclass[runningheads,a4paper]{llncs}
\usepackage{a4wide, amsmath, amssymb}
\usepackage{algorithmic}
\usepackage{epsfig}

\usepackage{algorithm}
\usepackage{amsfonts}

\newcommand{\Gr}{Gr\"obner }

\def\LM{{\mathrm{LM}}}
\def\LC{{\mathrm{LC}}}
\def\LT{{\mathrm{LT}}}

\def\poly{{\mathrm{poly}}}
\def \LT{{\rm LT}}
\def \LM{{\rm LM}}
\def \LC{{\rm LC}}

\def \lcm{{\rm lcm}}

\def\ri{\rangle}
\def\li{\langle}
\def\LM{{\mathrm{LM}}}
\def\LC{{\mathrm{LC}}}
\def\LT{{\mathrm{LT}}}
\def\L{{\mathcal{L}}}

\def\nf{{\rm{NF}}}
\def\poly{{\rm{poly}}}
\def\anc{{\rm{anc}}}

\def\NF{{\mathrm{NF}}}

\newcommand{\lex}{\mathop{\mathrm{lex}}\nolimits}

\newcommand{\degrevlex}{\mathop{\mathrm{degrevlex}}\nolimits}

\newenvironment{algorithm1}[1]
{
    \begin{center}
        {\bf Algorithm #1} \\
    \begin{tabular}{|p{140mm}|} \hline
}
{
    \\ \hline
    \end{tabular}
    \end{center}
}

\newenvironment{subalgorithm1}[1]
{
    \begin{center}
        {\bf Subalgorithm #1} \\
    \begin{tabular}{|p{140mm}|} \hline
}
{
    \\ \hline
    \end{tabular}
    \end{center}
}

\begin{document}

\mainmatter              
\title{Comprehensive Involutive Systems}
\titlerunning{Comprehensive Involutive Systems}
\author{
Vladimir Gerdt\inst{1} and Amir Hashemi \inst{2}}
\authorrunning{Vladimir Gerdt and Amir Hashemi}
\institute{Laboratory of Information Technologies, Joint Institute for Nuclear Research \\ 141980 Dubna, Russia, \ \
 e-mail: \email{gerdt@jinr.ru}
\and
Department of Mathematical Sciences, Isfahan University of Technology\\
Isfahan, 84156-83111, Iran, \ \
 e-mail: \email{Amir.Hashemi@cc.iut.ac.ir}
}

\maketitle              
\tocauthor{
Vladimir Gerdt (JINR, Dubna)
Amir Hashemi (Isfahan University of Technology, Isfahan)}

\begin{abstract}
In this paper we consider parametric ideals and introduce a notion
of {\em comprehensive involutive system}. This notion plays the
same role in theory of involutive bases as the notion of
comprehensive Gr\"obner system in theory of \Gr bases. Given a
parametric ideal, the space of parameters is decomposed into a
finite set of cells. Each cell yields the corresponding involutive
basis of the ideal for the values of parameters in that cell.
Using the Gerdt--Blinkov algorithm described in \cite{gerdtnew}
for computing involutive bases and also the Montes {\sc DisPGB}
algorithm for computing comprehensive Gr\"obner systems
\cite{monts1}, we present an algorithm for construction of
comprehensive involutive systems. The proposed algorithm has been
implemented in {\tt Maple}, and we provide an illustrative example
showing the step-by-step construction of comprehensive involutive
system by our algorithm.
\end{abstract}


\section{Introduction}
One of the  most important algorithmic objects  in computational
algebraic geometry is {\em Gr\"{o}bner basis}. The notion of
Gr\"obner basis was introduced and an algorithm for its
construction was designed in 1965 by Buchberger in his Ph.D.
thesis \cite{Bruno1}. Later on, he discovered \cite{Bruno2} two
criteria for detecting some useless reductions that made the
Gr\"{o}bner bases method a practical tool to solve a wide class of
problems in polynomial ideal theory and in many other research
areas of science and engineering~\cite{Bruno3}. We refer to the
monograph \cite{Becker} for details on the theory of \Gr bases.

The concept of comprehensive Gr\"obner bases can be considered as
an extension of these bases for polynomials over fields to
polynomials with parametric coefficients. This extension plays an
important role in application to constructive algebraic geometry,
robotics, electrical network, automatic theorem proving and so on
(see, for example, \cite{manmon1,manmon2,monts1,monts2}). {\em
Comprehensive Gr\"obner bases} and equivalent to them {\em
comprehensive Gr\"obner systems} were introduced  in 1992 by
Weispfenning \cite{Weis}. He proved that any parametric polynomial
ideal has a comprehensive Gr\"obner basis and described an
algorithm to compute it. In 2002, Montes \cite{monts1} proposed a
more efficient algorithm ({\sc DisPGB}) for computing
comprehensive Gr\"obner systems. A year later Weispfenning in
\cite{canWeis} proved the existence of a canonical comprehensive
Gr\"obner basis. In 2003, Sato and Suzuki \cite{sasu} introduced
the concept of alternative comprehensive Gr\"obner bases. Then in
2006, Manubens and Montes in \cite{manmon1} by using discriminant
ideal improved {\sc DisPGB}, and in \cite{manmon2} they introduced
an algorithm for computing minimal canonical Gr\"obner systems.
Also in 2006, Sato and Suzuki \cite{susa} (see also \cite{suzuki})
suggested an important computational improvement for comprehensive
Gr\"obner bases by constructing the reduced Gr\"obner bases in
polynomial rings over ground fields. In 2010, Kapur, Sun and Wang
\cite{kapur}, by combining Weispfenning's algorithm \cite{Weis}
with Suzuki and Sato's algorithm \cite{susa}, proposed a new
algorithm for computing comprehensive Gr\"obner systems. More
recently, in 2010, Montes and Wibmer in \cite{monts3}  presented
the {\sc Gr\"obnerCover} algorithm (its implementation in {\tt
Singular} is available at {\tt
http://www-ma2.upc.edu/$\sim$montes/}) which computes a finite
partition of the parameter space into locally closed subsets
together with polynomial data and such that the reduced Gr\"obner
basis for given values of parameters can immediately be determined
from the partition.

{\em Involutive bases} form an important class of \Gr bases. The
theory of involutive bases goes back to the seminal works of
French mathematician Janet. In the $20$s of the last century, he
developed \cite{janet} a constructive approach to analysis of
certain systems of partial differential equations based on their
completion to involution (cf.  \cite{seiler}). Inspired by the
involution methods described in the book by
Pommaret~\cite{pommaret}, Zharkov and Blinkov \cite{zharkov}
introduced the concept of {\em involutive polynomial bases} in
commutative algebra in the full analogy with the concept of
involutive systems of homogeneous linear partial differential
equations with constant coefficients and in one dependent
variable. Besides, Zharkov and Blinkov designed the first
algorithm for construction of involutive polynomial bases. The
particular form of an involutive basis they used is nowadays
called {\em Pommaret basis}~\cite{seiler}.

Gerdt and Blinkov \cite{gerdt0} proposed a more general concept of
involutive bases for polynomial ideals and designed efficient
algorithmic methods to construct such bases. The underlying idea
of the involutive approach is to translate the methods originating
from Janet's approach into the polynomial ideals theory in order
to provide a method for construction of involutive bases by
combining algorithmic ideas in the theory of Gr\"obner bases with
constructive ideas in the theory of involutive differential
systems. In doing so, Gerdt and Blinkov \cite{gerdt0} introduced
the concept of {\em involutive division}. Moreover, they derived
the involutive form of Buchberger's criteria. This led to a strong
computational tool which is a serious alternative to the
conventional Buchberger algorithm. We refer to Seiler's book
\cite{seiler} for a comprehensive study and application of
involution to commutative algebra and geometric theory of partial
differential equations.

In this paper, we introduce a notion of {\em comprehensive
involutive systems}. For a parametric ideal, we decompose the
space of parameters into a finite set of cells, and for each cell
we yield the corresponding involutive basis of the ideal. Thereby,
for each values of parameters, we find first a cell containing
these values. Then, by substituting these values into the
corresponding basis, we get the involutive basis of the given
ideal. Based on the Gerdt--Blinkov involutive (abbreviated below
by {\sc GBI}) algorithm as described in \cite{gerdtnew} and also
the Montes {\sc DisPGB} algorithm \cite{monts1}, we present an
algorithm for constructing comprehensive involutive systems. The
proposed algorithm has been implemented in {\tt Maple}, and we
provide an illustrative example showing the step-by-step results
of the algorithm.

The paper is structured as follows. Section \ref{WSS} contains the
basic definitions and notations related to comprehensive Gr\"obner
systems, and a short description of the {\sc DisPGB} algorithm.
The basic definitions and notations from the theory of involutive
bases are given in Section \ref{sec:1}. In Section \ref{gerdt},
the notion of comprehensive involutive system is introduced, and
an algorithm for construction of such systems is described. In
Section \ref{ex}, we give an example illustrating in detail the
performance of the algorithm of Section \ref{gerdt}.

\section{Comprehensive Gr\"obner Systems}
\label{WSS} In this section, we recall the basic definitions and
notations in theory of comprehensive Gr\"obner systems and briefly
describe the {\sc DisPGB} algorithm.

Let $R=K[{\bf x}]$ be a polynomial ring, where ${\bf
x}=x_1,\ldots,x_n$ is a sequence of variables and $K$ is an
arbitrary field. Below, we denote a monomial $x_1^{\alpha_1}\cdots
x_n^{\alpha_n}\in R$ by ${\bf x}^\alpha$ where
$\alpha=(\alpha_1,\ldots,\alpha_n) \in \mathbb{N}^{n}$ is a
sequence of non-negative integers. We shall use the notations
$\deg_i({\bf x}^\alpha):=\alpha_i$, $\deg({\bf
x}^\alpha):=\sum_{i=1}^n \alpha_i$. An {\em admissible} monomial
ordering on $R$ is a total order $\prec$ on the set of all
monomials such that for any $\alpha,\beta,\gamma \in
\mathbb{N}^{n}$ the following holds:

\[
  {\bf x}^\alpha \succ {\bf x}^\beta \Longrightarrow {\bf x}^{\alpha+\gamma} \succ {\bf x}^{\beta+\gamma},\qquad {\bf x}^\alpha\neq 1 \Longrightarrow {\bf x}^\alpha \succ 1\,.
\]

A typical example of admissible monomial ordering is the {\em lexicographical ordering}, denoted by $\prec_{\lex}$. If ${\bf x}^\alpha,{\bf x}^\beta\in R$ are two monomials, then ${\bf x}^\alpha \prec_{\lex} {\bf x}^\beta$ if the leftmost nonzero entry of $\beta-\alpha$ is positive. Another typical example is the {\em degree-reverse-lexicographical ordering} denoted by $\prec_{\degrevlex}$ and defined as ${\bf x}^\alpha \prec_{\degrevlex} {\bf x}^\beta$ if $\deg({\bf x}^\alpha)>\deg({\bf x}^\beta)$ or $\deg({\bf x}^\alpha)=\deg({\bf x}^\beta)$ and the rightmost nonzero entry of $\beta-\alpha$ is negative.

We shall write $I=\langle f_1,\ldots ,f_k\rangle$ for the ideal $I$ in $R$ generated by the polynomials $f_1,\ldots ,f_k\in R$. Let $f\in R$ and $\prec$ be a monomial ordering on $R$. The {\em leading monomial} of $f$ is the largest monomial (with respect to $\prec$) occurring in $f$, and we denote it by $\LM(f)$. If $F\subset R$ is a set of polynomials, then we denote by $\LM(F)$ the set $\{\LM(f) \ \mid \ f\in F\}$ of its leading monomials. The {\em leading coefficient} of $f$, denoted by $\LC(f)$, is the coefficient of $\LM(f)$. The {\em leading term} of $f$ is $\LT(f)=\LC(f)\LM(f)$. The {\em leading term ideal} of $I$ is
defined as $\LT(I)=\langle \LT(f)\ \mid \ f \in I\rangle$.

\noindent
A finite set $G=\{g_1,\ldots ,g_k\}\subset I$ is called a {\em Gr\"obner basis} of $I$ if $\LT(I)=\langle \LT(g_1),\ldots,\LT(g_k) \rangle$. For more details and definitions related to \Gr bases we refer to \cite{Becker}.

Now consider $F=\{f_1,\ldots,f_k\}\subset S:=K[{\bf a},{\bf x}]$ where ${\bf a}=a_1,\ldots,a_m$ is a sequence of parameters. Let $\prec_{{\bf x}}$ (resp. $\prec_{{\bf a}}$) be a monomial ordering for the power products of $x_i$'s (resp. $a_i$'s). We also need a compatible {\it elimination product ordering} $\prec_{{\bf x,a}}$. It is defined as follows: For all $\alpha,\gamma\in {\Bbb Z}^n_{\geq 0}$ and $\beta,\delta\in {\mathbb Z}^m_{\geq 0}$

\[
{\bf x}^\gamma {\bf a}^\delta\prec_{\bf{x,a}} {\bf x}^\alpha {\bf a}^\beta \Longleftrightarrow
{\bf x}^\gamma \prec_{\bf x} {\bf x}^\alpha  \ \text{or} \ {\bf x}^\gamma ={\bf x}^\alpha \ {\rm and} \ {\bf a}^\delta\prec_{\bf a} {\bf a}^\beta\,.
\]

Now, we recall the definition of a comprehensive Gr\"obner system for a parametric ideal.

\begin{definition}{\em(\cite{Weis})}
A triple set $\{(G_i,N_i,W_i)\}_{i=1}^{\ell}$ is called a {\em comprehensive Gr\"obner system} for $\langle F\rangle$ w.r.t $\prec_{{\bf x,a}}$ if for any $i$ and any homomorphism $\sigma:K[{\bf a}]\rightarrow K^\prime$ (where $K'$ is a field extension of $K$) satisfying
\[
    (i)\ (\,\forall p \in N_i\subset K[{\bf a}]\,)\ [\,{\sigma} (p)=0\,],\qquad
    (ii)\ (\,\forall q \in W_i\subset K[{\bf a}]\,)\ [{\sigma} (q)\ne 0\,]
\]
we have ${\sigma}(G_i)$ is a Gr\"obner basis for ${\sigma}(\langle F\rangle) \subset K'[{\bf x}]$ w.r.t. $\prec_{\bf x}$.
\end{definition}

For simplification, we shall use the abbreviation CGS to refer to a comprehensive Gr\"obner system, and CGSs in the plural case.  For each $i$, the set $N_i$ (resp.  $W_i$) is called a (resp. non-) null conditions set. The pair $(N_i, W_i)$ is called a {\em specification} of the homomorphism $\sigma$ if both conditions in the above definition are satisfied.

\begin{example}
\label{ex1}
  Let $F=\{ax^2y-y^3,bx+y^2\}\subset K[a,b,x,y]$ where ${\bf a}=a,b$  and ${\bf x}=x,y$. Let us consider the lexicographical monomial ordering $b\prec_{\lex} a$ on the parameters and on the variables $y\prec_{\lex}x$ as well. Using the {\sc DisPGB} algorithm we can compute a CGS for $\langle F\rangle$ which is equal to

$\begin{array}{lll}
\hspace*{3.5cm}\ \{-b^2y^3+ay^5,bx+y^2\}\  &\ \{~\} &\{a,b\}\\
\hspace*{3.5cm}\ \{x^2y,y^2\}\  &\ \{b\}\ &\ \{a\}\\
\hspace*{3.5cm}\ \{y^3,bx+y^2\}\  &\ \{a\}\ &\ \{b\}\\
\hspace*{3.5cm}\ \{y^2\}\  &\ \{a,b\}\ \ &\ \{~\}\,.\\
\end{array}$

For instance, if $a=0,b=2$, then the third element of this system corresponds to this specialization. Therefore, $\{y^3,2x+y^2\}$ is a Gr\"obner basis for the ideal $\langle F\rangle|_{a=0,b=2}=\langle  -y^3,2x+y^2 \rangle$.
\end{example}

Remark that, by the above definition, a CGS is not unique for a
given parametric ideal, and one can find other partitions for the
space of parameters, and, therefore, other CGSs for the parametric
ideal.

Now, we briefly describe the Montes {\sc DisPGB} algorithm to
compute CGSs for parametric ideals (see \cite{monts1,manmon1}).
The main idea of {\sc DisPGB} is based on discussing the nullity
or not w.r.t. a given specification $(N, W)$ for the leading
coefficients of the polynomials appearing at each step (this
process is performed by the {\sc NewCond} subalgorithm). Let us
consider a set $F\subset S$ of parametric polynomials. Given a
polynomial $f\in F$ and a specification $(N, W)$, {\sc NewCond} is
called. Three cases are possible: If $\LC(f)$ specializes to zero
w.r.t. $(N,W)$, we replace $f$ by  $f - \LT(f)$, and then start
again. If $\LC(f)$ specializes to a nonzero element we continue
with the next polynomial in $F$. Otherwise (if $\LC(f)$ is not
decidable, i.e. we can't decide whether or not it is null w.r.t.
$(N,W)$), the subalgorithm {\sc Branch} is called to create two
complementary cases by assuming $\LC(f)=0$ and $\LC(f)\ne 0$.
Therefore,  two new disjoint branches with the specifications  $(N
\cup \{\LC(f)\},W)$ and $(N,W \cup \{\LC(f)\})$ are made. This
procedure is continued until every polynomial in $F$ has a nonnull
leading coefficient w.r.t. the current specification. Then, we
proceed with {\sc CondPGB}: This algorithm receives, as an input,
a set of parametric polynomials and a specification $(N,W)$ and,
by applying Buchberger's algorithm, creates new polynomials. When
a new polynomial is generated, {\sc NewCond} verifies whether its
leading coefficient leads to a new condition or not.  If a new
condition is found, then the subalgorithm stops, and {\sc Branch}
is called to make two new disjoint branches. Otherwise, the
process is continued and computes a Gr\"obner basis for $\li F
\ri$, according to the current specification. The collection of
these bases, together with the corresponding specifications yields
a CGS for  $\li F \ri$.

\section{Involutive Bases}
\label{sec:1}
Now we recall the basic definitions and notations concerning involutive bases and present below the general definition of involutive bases. First of all, we describe the cornerstone notion of {\em involutive division} \cite{gerdt0} as a restricted monomial division \cite{gerdtnew} which, together with a monomial ordering, determines properties of an involutive basis. This makes the main difference between involutive bases and Gr\"obner bases. The idea behind involutive division is to partition the variables into two subsets of multiplicative and nonmultiplicative variables, and only the multiplicative variables can be used in the divisibility relation.

\begin{definition}{\em\cite{gerdt0,gerdtnew}}
An {\em involutive division} $\L$ on the set of monomials of $R$ is given, if  for any finite set $U$ of monomials and any $u \in U$, the set of variables is partitioned into subsets $M_{\L}(u,U)$ of {\em multiplicative} and $NM_{\L}(u,U)$ of {\em nonmultiplicative} variables such that
\begin{enumerate}
\item $u,v\in U,\ u\L(u,U) \cap v\L(v,U) \ne \emptyset \Longrightarrow u\in v\L(v,U)$ or $v \in u\L(u,U)$,
\item $v\in U,\ v \in u\L(u,U) \Longrightarrow \L(v,U) \subset \L(u,U)$,
\item $u \in V$ and $V \subset U \Longrightarrow \L(u,U) \subset \L(u,V)$,
\end{enumerate}
where $\L(u,U)$ denotes  the set of all monomials in the variables in $M_\L(u,U)$. If $v \in u\L(u,U)$, then we call $u$ an $\L-${\em (involutive) divisor} of $v$, and we write $u |_\L v$. If $v$ has no involutive divisor in a set $U$, then it is {\em $\L-$irreducible} modulo $U$.
\end{definition}

In this paper, we are concerned with the wide class \cite{alex} of involutive divisions determined by a permutation $\rho$ on the indices of variables and by a total monomial ordering $\sqsupset$ which is either admissible or the inverse of an admissible ordering. This class is defined by
\begin{equation}
 (\ \forall u\in U\ ) \ \ [\ NM_{\sqsupset}(u,U)=\bigcup_{v\in U\setminus \{u\}}NM_{\sqsupset}(u,\{u,v\})\ ]
 \label{pair}
\end{equation}
where
\begin{equation}
NM_{\sqsupset}(u,\{u,v\}):=\left\lbrace
\begin{array}{l}
\text{ if } u\sqsupset v \text{ or } (u\sqsubset v \wedge v\mid u) \text{ then } \emptyset \\[0.1cm]
\text{ else }\{x_{\rho(i)}\},\ i=\min\{j\mid \deg_{\rho(j)}(u)< \deg_{\rho(j)}(v)\}\,.\\
\end{array}
\right. \label{inv_div}
\end{equation}

\begin{remark}
The involutive Janet division introduced and studied in
\cite{gerdt0} is generated by formulae
(\ref{pair})--(\ref{inv_div}) if  $\sqsupset$ is the lexicographic
monomial ordering $\succ_{\lex}$ and $\rho$ is the identical
permutation. The partition of variables used by Janet
himself~\cite{janet} (see also \cite{seiler}) is generated by
$\succ_{\lex}$ as well with the permutation which is inverse to
the identical one:
\begin{displaymath}
\rho=\left(
\begin{array}{cccc}
1 & 2 & \ldots & n \\
n & n-1 & \ldots & 1\\
\end{array}
\right)\,.
\end{displaymath}
\end{remark}

{\em Throughout this paper $\L$ is assumed to be a division of the
class (\ref{pair})--(\ref{inv_div})}. Now, we define an involutive
basis.
\begin{definition}
Let $I\subset R$ be an ideal, $\prec$ be a monomial ordering on $R$ and $\L$ be an involutive division. A finite set $G\subset I$ is an {\em involutive basis} of $I$ if for all $f\in I$ there exists $g\in G$ such that $\LM(g) |_\L \LM(f)$. An involutive basis $G$ is {\em minimal} if for any other involutive basis $\tilde{G}$ the inclusion $\LM(G)\subseteq \LM(\tilde{G})$ holds.
\label{def_ib}
\end{definition}

>From this definition and from that for Gr\"obner basis~\cite{Bruno1,Becker} it follows that an involutive basis of an ideal is its   Gr\"obner basis, but the converse is not always true.

\begin{remark}
By using an involutive division in the division algorithm for polynomial rings, we obtain an involutive division algorithm. If $G$ is an involutive basis for an involutive division $\L$, we use $\nf_\L(f,G)$ to denote {\em ${\L}-$normal form} of $f$ modulo $G$, i.e.
the remainder of $f$ on the involutive division by $G$. A polynomial set $F$ is {\em $\L-$autoreduced} if $f=\NF_{\L}(f,F\setminus \{f\})$ for every $f\in F$.
\end{remark}

\noindent
The following theorem provides an algorithmic characterization of involutive bases which is an involutive analogue of the Buchberger characterization of Gr\"obner bases.

\begin{theorem}{\em(\cite{gerdt0,alex})}
\label{blin}
Given an ideal $I\subset R$, an admissible monomial ordering $\prec$ on $R$ and an involutive division $\L$, a finite subset $G \subset I$ is an involutive basis of $I$ if for each $f\in G$ and each $x\in NM_\L(\LM(f),\LM(G))$ the equality $\nf_\L(xf,G)=0$ holds. An involutive basis exists for any $I$, $\L$ and $\prec$. A monic and $\L$-autoreduced involutive basis is uniquely defined by $I$ and $\prec$.
\label{alg_ib}
\end{theorem}

\section{Comprehensive Involutive Systems}
\label{gerdt}

In this section, like the concept of comprehensive Gr\"obner
systems, we define the new notion  of comprehensive involutive
system  for a parametric ideal. Then, based on the GBI algorithm
\cite{gerdtnew} and the Montes {\sc DisPGB} algorithm
\cite{monts1}, we propose an algorithm for computing comprehensive
involutive systems.

\begin{definition}
Consider a finite set of parametric polynomials $F\subset S=K[{\bf a},{\bf x}]$ where $K$ is a field, ${\bf x}=x_1,\ldots,x_n$ is a sequence of variables and ${\bf a}=a_1,\ldots,a_m$ is a sequence of parameters, $\prec_{{\bf x}}$ (resp. $\prec_{{\bf a}}$) is a monomial ordering involving the $x_i$'s (resp. $a_i$'s), and $\L$ is an involutive division on $K[{\bf x}]$. Let  $M=\{(G_i,N_i,W_i)\}_{i=1}^{\ell}$ be a finite triple set where sets $N_i,W_i\subset K[{\bf a}]$ and $G_i\subset S$ are finite. The set $M$ is called an {\em ($\L-$)comprehensive involutive  system} for $\li F\ri$ w.r.t $\prec_{{\bf x,a}}$ if for each $i$ and for each homomorphism $\sigma:K[a]\rightarrow K^\prime$ (where $K'$ is a field extension of $K$) satisfying
\[
     (i)\ (\,\forall p \in N_i\,)\ [\, {\sigma} (p)=0\,], \qquad (ii)\ (\,\forall q \in W_i\,)\ [\,{\sigma} (q)\ne 0 \,]
\]
${\sigma}(G_i)$ is an ($\L-$)involutive basis for ${\sigma}(\li F\ri)\subset  K'[{\bf x}]$. We use the abbreviation CIS (resp. CISs) to stand for comprehensive involutive  system (resp. systems). $M$ is called {\em minimal}, if for each $i$, the set ${\sigma}(G_i)$ is a minimal involutive basis.
\end{definition}

Given a CGS, one can straightforwardly compute a CIS by using the following proposition.

\begin{proposition}
\label{prop1}
Let $G=\{g_1,\ldots,g_k\}$ be a minimal Gr\"obner basis of an ideal $I\subset K[x_1,\ldots,x_n]$ for a monomial ordering $\prec$. Let $h_i=\max_{g\in G}\{\deg_i(\LM(g))\}$. Then the set of products
\begin{equation}
\{mg \mid g\in G,\, \ m \ {\rm{is \ a \ monomial \ s.t.}}\ (\,\forall i\,)\ [\,\deg_i(m)\le h_i-\deg_i(\LM(g))\,]\}
\label{T-compl}
\end{equation}
 is an $\L$-involutive basis of $I$.
\end{proposition}

\begin{proof}
Denote $\LM(G)$ by $U$. From (\ref{pair})--(\ref{inv_div}) it
follows
\begin{equation}
(\,\forall u\in U\,)\ (\,\forall x_i\in NM_{\L}(u,U))\,)\ [\,\deg_i(u)<h_i\,]\,.
\label{nm_inclusion}
\end{equation}
It is also clear that if we enlarge $G$ with a (not necessarily nonmultiplicative) prolongation $gx_j$ of its element $g\in G$ such that $\deg_j(\LM(g))<h_j$, then (\ref{nm_inclusion}) holds for the enlarged leading monomial set $U:=U\cup \{\LM(g)x_j\}$ as well. Consider  completion $\bar{G}$ of the polynomial set $G$ with all possible prolongations of its elements satisfying (\ref{T-compl}) and denote  the monomial set $\LM(\bar{G})$ by $\bar{U}$. Then
\[
(\,\forall u\in \bar{U}\,)\ (\,\forall x\in NM_{\L}(u,U)\,)\ (\,\exists v\in \bar{U}\,)\ [\,v\mid_{\L} ux\,]\,.
\]
This means, by Theorem \ref{alg_ib}, that the monomial set $\bar{U}$ is an involutive basis of $\langle \LM(G)\rangle$. Now, since $G$ is a \Gr basis of $I$ we have $\LT(I)=\langle \LM(G)\rangle$, and hence $\LT(I)=\langle \LM(\bar{G})\rangle$. Therefore, $\bar{G}$ is an involutive basis of $I$ by Definition \ref{def_ib}. $\Box$.
\end{proof}

\begin{example}
\label{ex2}
  Let $F=\{ax^2,by^2\}\subset \mathbb{K}[{\bf a},{\bf x}]$ where ${\bf a}=a,b$  and ${\bf x}=x,y$. Let also $b\prec_{lex} a$ and $y\prec_{lex}x$. Then, $F$ is a CGS for any sets of null and nonnull conditions. Using Proposition \ref{prop1}, we can construct the following Janet basis of $\li F\ri$ which is a GIS for any sets of null and nonnull conditions:
$$\{ax^2,by^2,ayx^2,ay^2x^2,bxy^2,bx^2y^2\}\,.$$
On the other hand, the algorithm that we present below computes the following minimal Janet CIS for $\li F\ri$:

$\begin{array}{lll}
\hspace*{3.5cm}\ \{ax^2,by^2,bxy^2\}\ \ & \{\ \} & \{a,b\}\\
\hspace*{3.5cm}\ \{ax^2\}  & \{b\} & \{a\}\\
\hspace*{3.5cm}\ \{by^2\}  & \{a\} & \{b\}\\
\hspace*{3.5cm}\ \{0\}  & \{a,b\}\ \ & \{~\}.\\
\end{array}$
\end{example}
\begin{remark}
  Using Proposition \ref{prop1}, we cannot directly compute a minimal CIS from a given CGS. Indeed, to do this, we must examine the leading coefficients of each  Gr\"obner basis in the CGS, and this may lead to further partitions of the space of parameters. Moreover, the CIS computed by this way may be too large, since many prolongations constructed by means of (\ref{T-compl}) may be useless. That is why, based on the GBI algorithm \cite{gerdtnew} and on the Montes {\sc DisPGB} algorithm \cite{monts1}, we propose a more efficient algorithm for computing minimal CISs.
\end{remark}

Now we describe the structure of polynomials that is used in our new algorithm. To avoid unnecessary reductions (during the computation of involutive bases) by applying the involutive form of Buchberger's criteria (see \cite{gerdtnew}), we need to supply polynomials with additional structural information.

\begin{definition}{\em\cite{gerdtnew}}
An {\em ancestor} of a polynomial $f \in F \subset R \setminus \{0\}$, denoted by $\anc(f)$, is a polynomial $g \in  F$
of the smallest $\deg(\LM(g))$ among those satisfying $\LM(f) = u\LM(g)$ where $u$ is either the unit monomial or a power product of nonmultiplicative variables for $\LM(g)$ and such that $\nf_{\L}(f-ug,F\setminus\{f\})=0$ if $f\neq ug$.
\end{definition}

Below we show how to use this concept to apply the involutive form of Buchberger's criteria.  In what follows, we store each polynomial $f$ as the $p=[f,g,V]$ where $f=\poly(p)$ is the polynomial part of $p$, $g=\anc(p)$ is the ancestor of $f$ and $V=NM(p)$ is the list of nonmultiplicative variables of $f$ have been already used to construct prolongations of $f$ (see the {\bf for-}loop 20-23 in the subalgorithm GBI).  If $P$ is a set of triples, we denote by $\poly(P)$ the set $\{\poly(p) \ \mid \ p\in P\}$. If no  confusion arises, we may refer to a triple $p$ instead of $\poly(p)$, and vice versa.

We present now the main algorithm {\sc ComInvSys} which computes a minimal CIS for a given ideal. It should be noted that we use the subalgorithms {\sc NewCond}  and {\sc CanSpec} (resp. {\sc TailNormalForm}) as they have (resp. it has) been presented in \cite{monts1} (resp. \cite{gerdtnew}), and recall them for the sake  of completeness. Also, we use the subalgorithm {\sc Branch} (resp. {\sc GBI} , {\sc HeadReduce} and {\sc HeadNormalForm}) from \cite{monts1} (resp. \cite{gerdtnew}) with some appropriate modifications.

\begin{algorithm1}{{\sc ComInvSys}
\label{ComInvSys}}
\begin{algorithmic}[1]
\INPUT $F$, a set of polynomials; $\L$, an involutive division; $\prec_{{\bf x}}$,  a monomial ordering on the variables; $\prec_{{\bf a}}$, a monomial ordering on the parameters
\OUTPUT a minimal CIS for $\langle F\rangle$
    \STATE global: {\tt List}, {\tt ind};
    \STATE {\tt List}:=Null;
    \STATE {\tt ind}:=1;
    \STATE $B:=\{[F[i],F[i],\emptyset] \ | \ i=1,\ldots ,|F| \}$;
    \STATE $G:=\{${\sc Branch}$([F[1],F[1],\emptyset],B,\{~\},\{~\},\{~\})\}$;
    \FOR {i {\bf from} $2$ {\bf to} $|F|$}
         \STATE {\tt ind}:={\tt ind}$+1$;
         \STATE $G:=\{${\sc Branch}$([F[i],F[i],\emptyset],A[2],A[3],A[4],A[5]) \ | \ A \in G\}$;
    \ENDFOR
    \STATE {\bf Return} ({\tt List})
\end{algorithmic}
\end{algorithm1}

\noindent
In the above algorithm, {\tt List} is a global variable to which we add any computed involutive basis together with its corresponding specification to form the final CIS. That is why, at the beginning of computation we must set it to the empty list (see {\sc Branch}). Note that here and in {\sc Branch}, we use $|F|$ to denote the number of polynomials in the input set $F$. The variable {\tt ind} is also a global variable, and we use it to examine all the leading coefficients of the elements in $F$ (see {\sc Branch}). Once we are sure about the non-nullity of these coefficients, then we start the involutive basis computation. Indeed, {\sc Branch} inputs a triple $p=[f,g,V]$, a set $B$ of examined and processed polynomials, a set $N$ of null conditions, a set $W$ of nonnull conditions and a set $P$ of non-processed polynomials. Then, it analyses the leading coefficient of $f$ w.r.t. $N$ and $W$. Now, two cases are possible:
\begin{itemize}
  \item {\tt ind}$<|F|$: If $\LC(f)$ is not decidable by $N$ and $W$ then we create two complementary cases by assuming $\LC(f)=0$ and $\LC(f)\ne 0$. Then we pass to the next polynomial in $F$.
\item {\tt ind}$=|F|$: We are now sure that we have examined all the leading coefficients of the elements in $F$ (except possibly the very last one which is to be $f$). If $\LC(f)$ is not decidable by $N$ and $W$ then we again create two complementary cases with $\LC(f)=0$ and $\LC(f)\ne 0$. Otherwise, we continue to process the polynomials in $P$ by using the {\sc GBI}  algorithm. If $P=\emptyset$ this means that $B$ is an involutive basis  consistent with the conditions in $N$ and $W$, and we add $(B,N,W)$ to {\tt List}.
\end{itemize}

It is worth noting that if the input specification of {\sc Branch} is incompatible, then it stops the process only for the corresponding  branch, and  continues the construction of other branches. Moreover, using the above notations, if {\tt ind}$<|F|$ and no new condition is detected, then {\sc Branch} returns an element of the form $(p,{\overline{B}}^{N'},N',W',P)$ where  $p$ is a triple, $N',W'$ are two sets of conditions, ${\overline{B}}^{N'}$  is the normal form of  a specializing basis $B$ and $P$ is a set of non-examined triples. Otherwise, it calls itself to create the new branches. Finally, if {\tt ind}$=|F|$, then the algorithm does not return anything and completes the global variable {\tt List}.

\begin{subalgorithm1}{{\sc Branch}
\label{Branch}}
\begin{algorithmic}[1]
\INPUT $p$, a triple; $B$, a specializing basis; $N$, a set of null conditions; $W$, a set of nonnull conditions; $P$, a set of non-examined triples
\OUTPUT It stores the refined $(B', N', W',P')$, and creates two new vertices when necessary or marks the vertex as terminal
\STATE $p=[f,g,V]$;
\STATE {$(test,N,W)$:={\sc CanSpec}$(N,W)$};
\IF{$test$=false}
 \STATE {\bf Return} STOP (incompatible specification has been detected)
\ENDIF
\STATE $(cd,f',N',W'):=${\sc NewCond}$(f,N,W)$;
\STATE $p:=[f',{\overline{g}}^{N'},V]$ (${\overline{g}}^{N'}$ denotes the remainder of the division of $g$ by $N'$);
\IF{{\tt ind} $<|F|$ and $cd \ne  \{~\}$}
     \STATE  {\sc Branch}$(p,B,N',W'\cup cd,P)$;\quad {\sc Branch}$(p,B,N'\cup cd,W',P)$;
\ENDIF
\IF{{\tt ind} $<|F|$ and $cd =  \{~\}$}
     \STATE {\bf Return} $(p,{\overline{B}}^{N'},N',W',P)$
\ENDIF
\IF{$cd = \{~\}$}
    \STATE $(test,p',B',N',W',P'):=${\sc GBI} $(B,N',W',P)$;
    \IF {\em test}
            \STATE {\tt List}:={\tt List},$(B',N',W')$;
     \ELSE
        \STATE {\sc Branch}$(p',B',N',W',P')$;
     \ENDIF
\ELSE
     \STATE  {\sc Branch}$(p,B,N',W'\cup cd,P)$;\quad {\sc Branch}$(p,B,N'\cup cd,W',P)$;
\ENDIF
\end{algorithmic}
\end{subalgorithm1}

\noindent The subalgorithm {\sc CanSpec} produces a
quasi-canonical representation for a given specification.  Its
subalgorithm {\sc FacVar} invoked in lines 1 and  13 returns the
set of factors of its input polynomial.

\begin{definition}{\em(\cite{monts1})}
 A specification $(N,W)$ is called {\it quasi-canonical} if
\begin{itemize}
\item $N$ is the reduced Gr\"obner basis w.r.t. $\prec_{\bf a}$ of the ideal containing all  polynomials that specialize to zero in $K[{\bf a}]$.
\item The polynomials in $W$ specializing to non-zero are reduced modulo $N$ and irreducible over $K[{\bf a}]$
\item $\prod_{q\in W}q \notin \sqrt {\li N \ri}$.
\item  The polynomials in $N$ are square-free  over $K[{\bf a}]$.
\item  If some $p \in N$ is factorized, then no factor of $p$ belongs to $W$.
\end{itemize}
\end{definition}

\vskip 0.2cm
\begin{subalgorithm1}{{\sc CanSpec}
\label{CanSpec}}
\begin{algorithmic}[1]
\INPUT $N$, a set of null conditions; $W$, a set of nonnull conditions
\OUTPUT true  if $N$ and $W$  are compatible and false otherwise;
   $(N',W')$, a quasi-canonical representation of $(N,W)$
 \STATE $W':= ${\sc FacVar}$(\{{\overline{q}}^N : q \in W\})$;\quad
       $test$:=true;\quad $N':=N$;\quad $h:=\prod_{q\in W} q$;
 \IF {$h \in \sqrt{\li N'\ri}$}
      \STATE $test$:=false;\quad $N':=\{1\}$;
      \STATE {\bf Return} $(test , N',W')$;
 \ENDIF
 \STATE $flag$:=true;
 \WHILE {$flag$}
     \STATE $flag$:=false;
     \STATE $N''$:= Remove any factor of a polynomial in $N'$ that belongs to $W'$;
     \IF {$N''\ne N'$}
          \STATE $flag$:=true;
          \STATE $N'$:= a Gr\"obner basis of $\li N'' \ri$ w.r.t. $\prec_{\bf a}$;
          \STATE  $W':= ${\sc FacVar}$(\{{\overline{q}}^{N^{'}} : q \in W'\})$;
      \ENDIF
 \ENDWHILE
\STATE {\bf Return} $(test , N',W')$
\end{algorithmic}
\end{subalgorithm1}

\vskip 0.3cm
\begin{subalgorithm1}{{\sc NewCond}
\label{NewCond}}
\begin{algorithmic}[1]
\INPUT $f$, a parametric polynomial; $N$, a set of null conditions; $W$, a set of nonnull conditions
\OUTPUT $cd$, a new condition; $f'$, a parametric polynomial;
$N'$, a set of null conditions; $W'$, a set of nonnull conditions
 \STATE $f':=f$;\quad $test$:=true;\quad $N':=N$;\quad cd:=\{\ \};
 \WHILE {$test$}
     \IF {$\LC(f') \in \sqrt{\li N' \ri}$}
          \STATE $N':=$ a Gr\"obner basis for $\li N',\LC(f')\ri$ w.r.t. $\prec_{\bf a}$;
          \STATE $f':=f'-\LT(f)$;
     \ELSE
          \STATE $test$:=false;
     \ENDIF
\ENDWHILE
\STATE $f':={\overline{f'}}^{N'}$;
\STATE $W':=\{{\overline{w}}^{N'} \ |\ w\in W\}$;
\STATE $cd:=cd\,\, \cup\,${\sc FacVar}$(\LC(f')) \setminus W';$
\STATE {\bf Return}$(cd,f',N',W')$
\end{algorithmic}
\end{subalgorithm1}

\noindent
We describe now the {\sc NewCond} subalgorithm. When it is invoked in line 6 of {\sc Branch} with the input data $(f,N,W)$, one of the two following cases may occur:
 \begin{enumerate}
   \item If $\LC(f)$ is decidable w.r.t. the specification $(N,W)$, then the subalgorithm returns:
     \begin{itemize}
       \item[(i)] {\sc NewCond}$(f-\LT(f),N,W)$ in the case when $\LC(f)$ specializes to zero w.r.t. $(N,W)$.
        \item[(ii)] ($\emptyset,f,N,W)$ in the case when $\LC(f)$ does not specialize to zero w.r.t. $(N,W)$.
     \end{itemize}
    \item If $\LC(f)$ is not decidable w.r.t $(N,W)$, then {\sc NewCond} returns $(cd,f,N,W)$ where set $cd$ contains one of the non-decidable factors (w.r.t $(N,W)$) of $\LC(f)$.
 \end{enumerate}
 It should be emphasized that {\sc FacVar}$(\LC(f')) \setminus W'$  in line 12 returns only one factor of $\LC(f')$.

The subalgorithm {\sc GBI}, presented below, is an extension of the algorithm {\sc InvolutiveBasis II} described in \cite{gerdtnew}. The latter algorithm computes involutive bases and applies the involutive form of Buchberger's criteria to avoid some unnecessary reductions~\cite{gerdt0} (see also \cite{detecting,gerdtnew}). The criteria are applied in the subalgorithm {\sc HeadNormalForm} (see line 7) that is invoked in line 5 of {\sc GBI}.

\begin{proposition}{\em(\cite{gerdt0,gerdtnew})}
\label{crit}
Let $I\subset R$ be an ideal and $G\subset I$ be a finite set. Let also $\prec$ be a monomial ordering on $R$ and $\L$ be an involutive division. Then $G$ is an $\L-$involutive basis of $I$ if for all $f\in G$ and for all $x\in NM_\L(\LM(f),\LM(G))$ one of the two conditions holds:
\begin{enumerate}
\item $\nf_\L(xf,G)=0$\,.
\item There exists $g \in G$ with $\LM(g) |_\L \LM(xf)$ satisfying one of the following conditions:
\begin{description}
\item[$(C_1)$] $\LM(\anc(f))\LM(\anc(g))=\LM(xf)$\,,
\item[$(C_2)$] $\lcm(\LM(\anc(f)),\LM(\anc(g)))$ is a proper divisor of $\LM(xf)$\,.
\end{description}
\end{enumerate}
\end{proposition}

\begin{subalgorithm1}{{\sc GBI}
\label{Gerdt}}
\begin{algorithmic}[1]
\INPUT $B$, a specializing basis; $N$, a set of null conditions;
$W$, set of nonnull conditions; $P$, set of non-examined triples
\OUTPUT If $test$=true, a minimal involutive basis for $\li B \ri$ w.r.t. $\L$ and $\prec_{{\bf x,a}}$; otherwise, it returns a triple so that we must discuss the leading coefficient of its polynomial part
\IF {$P=\emptyset$}
     \STATE Select $p \in B$ with no proper divisor of $\LM(\poly(p))$ in $\LM(\poly(B))$
     \STATE $T:=\{p\}$;\quad $Q:=B \setminus \{p\}$;
\ELSE
     \STATE $T:=B$;\quad $Q:=P$;
\ENDIF
\WHILE {$Q \ne \emptyset$}
     \STATE $(test,p,T,N,W,Q'):=${\sc HeadReduce}$(T,N,W,Q)$;
     \IF {$test=$false}
           \STATE {\bf Return} $({\rm false},p,T,N,W,Q')$
     \ENDIF
     \STATE $Q:=Q'$;
     \STATE Select and remove $p \in Q$ with no proper divisor of $\LM(\poly(p))$ in $\LM(\poly(Q))$;
     \IF {$\poly(p)=\anc(p)$}
           \FOR {$q\in T$ whose $\LM(\poly(q))$ is a proper multiple of $\LM(\poly(p))$}
                 \STATE $Q:=Q \cup \{q\}$;\quad $T:=T \setminus \{q\}$;
           \ENDFOR
     \ENDIF
      \STATE $h:=${\sc TailNormalForm}$(p,T)$;\quad $T:=T \cup \{\{h,\anc(p),NM(p)\}\}$;
      \FOR {$q\in T$ and $x\in NM_\L(\LM(\poly(q)),\LM(\poly(T))) \setminus NM(q)$}
           \STATE $Q:=Q \cup \{\{x\ \poly(q),\anc(q),\emptyset\}\}$;
           \STATE $NM(q):=NM(q) \cap NM_\L(\LM(\poly(q)),\LM(\poly(T))) \cup \{x\}$;
      \ENDFOR
\ENDWHILE
\STATE {\bf Return} (${\rm true},0,T,N,W,\{~\}$)
\end{algorithmic}
\end{subalgorithm1}

\noindent
This algorithm invokes  three subalgorithms {\sc HeadReduce}, {\sc TailNormalForm} and {\sc HeadNormalForm} that we present below. The subalgorithm {\sc HeadReduce} performs the involutive head reduction of polynomials in the input set of triples modulo the input specializing basis. The subalgorithm {\sc TailNormalForm} (resp. {\sc HeadNormalForm}) invoked in line 19 of {\sc GBI} (resp. in line 4 of {\sc HeadReduce}) computes the involutive tail normal form (resp. the involutive head normal form) of the polynomial in the input triple modulo the input specializing basis.

In the subalgorithm {\sc HeadNormalForm}, the Boolean expression Criteria$(p,g)$ is true if at leat one of the conditions $(C_1)$ or $(C_2)$ in Proposition \ref{crit} are satisfied for $p$ and $g$, false otherwise. We refer to \cite{gerdtnew} for more details on the algorithm {\sc GBI} and on its subalgorithms.

\begin{subalgorithm1}{{\sc HeadReduce}
\label{HeadReduce}}
\begin{algorithmic}[1]
\INPUT $B$, a specializing basis; $N$,  a set of null conditions; $W$, a set of nonnull conditions;
 $P$ a set of non-examined triples
\OUTPUT If $test$=true, the $\L$-head reduced form of $P$ modulo $B$; otherwise, it returns a triple such that we must examine the leading coefficient of its polynomial part
 \STATE $S:=P$;\quad $Q:=\emptyset$;
\WHILE {$S \ne \emptyset$}
     \STATE Select and remove $p \in S$;
     \STATE $(test,h,B,N,W):=${\sc HeadNormalForm}$(p,B,N,W)$;
\IF {$test$=false}
    \STATE {\bf Return} $({\rm false},p,B,N,W,S\cup Q)$
\ENDIF
\IF {$h\ne 0$}
          \IF{$\LM(\poly(p))\ne \LM(h)$}
               \STATE $Q:=Q \cup \{\{h,h,\emptyset\}\};$
          \ELSE
                \STATE $Q:=Q \cup \{p\};$
          \ENDIF
     \ELSE
          \IF{$\LM(\poly(p))=\LM(\anc(p))$}
               \STATE $S:=S \setminus \{q\in S \ | \ \anc(q)=\poly(p)\}$;
          \ENDIF
      \ENDIF
\ENDWHILE
\STATE {\bf Return} (${\rm true},0,B,N,W,Q$)
\end{algorithmic}
\end{subalgorithm1}

\vskip 0.5cm
\begin{subalgorithm1}{{\sc TailNormalForm}
\label{TailNormalForm}}
\begin{algorithmic}[1]
\INPUT $p$, a triple; $B$, a set of triples
\OUTPUT $\L$-normal form of $\poly(p)$ modulo $\poly(B)$
 \STATE $h:=\poly(p)$;
\STATE $G:=\poly(B)$;
\WHILE {$h$ has a term $t$ which is  $\L-$reducible modulo $G$}
       \STATE Select $g\in G$ with $\LM(g) |_\L t$;
       \STATE $h:=h-g\frac{t}{\LT(g)}$;
\ENDWHILE
\STATE {\bf Return} ($h$)
\end{algorithmic}
\end{subalgorithm1}

\newpage
\begin{subalgorithm1}{{\sc HeadNormalForm}
\label{HeadNormalForm}}
\begin{algorithmic}[1]
\INPUT $p$, a triple; $B$, a specializing basis; $N$, a set of null conditions; $W$, set of nonnull conditions
\OUTPUT If $test$=true, the $\L$-head normal form of $\poly(p)$ modulo $B$; otherwise, a polynomial whose leading coefficient must be examined
 \STATE $h:=\poly(p)$;\quad $G:=\poly(B)$;
\IF{$\LM(h)$ is $\L$-irreducible modulo $G$}
     \STATE {\bf Return} (${\rm true},h,B,N,W$)
\ELSE
     \STATE Select $g\in G$ with $\LM(\poly(g)) |_\L \LM(h)$;
          \IF{$\LM(h)\ne \LM(\anc(p))$}
               \IF{Criteria$(p,g)$}
                     \STATE {\bf Return} (${\rm true},0,B,N,W$)
               \ENDIF
          \ELSE
               \WHILE {$h\ne 0$ and $\LM(h)$ is $\L$-reducible modulo $G$}
                     \STATE Select $g\in G$ with $\LM(g) |_\L \LM(h)$;
                     \STATE $h:=h-g\frac{\LT(h)}{\LT(g)}$;
                     \STATE $(cd,h',N',W'):=${\sc NewCond}$(h,N,W)$;
                     \IF{$cd\ne \emptyset$}
                            \STATE {\bf Return} (${\rm false},h',B,N',W'$)
                     \ENDIF
               \ENDWHILE
           \ENDIF
\ENDIF
\STATE {\bf Return} (${\rm true},h,B,N,W$)
\end{algorithmic}
\end{subalgorithm1}

\begin{theorem}
  Algorithm {\sc ComInvSys} terminates in finitely many steps, and computes a minimal {\rm CIS} for its input ideal.
\end{theorem}

\begin{proof}
Let $I=\li F \ri$ where $F=\{f_1,\ldots, f_k\}\subset K[{\bf a},{\bf x}]$ is a parametric set, ${\bf x}=x_1,\ldots,x_n$ (resp. ${\bf a}=a_1,\ldots,a_m$) is a sequence of variables (resp. parameters).  Let $\prec_{{\bf x}}$ (resp. $\prec_{{\bf a}}$) be a monomial ordering involving the $x_i$'s (resp. $a_i$'s), and $\L$ be an involutive division on $K[{\bf x}]$.

Suppose that {\sc ComInvSys} receives $F$ as an input. To prove the {\em termination}, we use the fact  that $K[{\bf a}]$ is a Noetherian ring. When {\sc Branch} is called, the leading coefficient of some polynomial $f\in I$ is analyzed. For this purpose, the subalgorithm {\sc NewCond} determines whether $\LC(f)$ is decidable or not w.r.t. the given specification $(N,W)$. Two alternative cases can take place:
\begin{itemize}
  \item $\LC(f)$ is decidable and we check the global variable {\tt ind}. Now if {\tt ind}$<k$, then we study the next polynomial in $F$. Otherwise, {\sc GBI}  is called. If all the leading coefficients of the examined polynomials (to compute a minimal involutive basis) are decidable, then the output, say  $G$, is a minimal involutive basis of $I$ w.r.t. $(N,W)$, and we add $(G,N,W)$ to {\tt List}. Otherwise, two new branches are created by calling {\sc Branch} (cf. the second case given below). In doing so, the minimality of $G$ and the termination of its computation is provided by the structure of {\sc GBI}  algorithm (see \cite{gerdtnew}).
   \item $\LC(f)$ is not decidable and we create two branches with $(N,W\cup cd)$ and $(N\cup cd,W)$,
   where $cd$ is the one-element set containing the new condition derived from $\LC(f)$.
\end{itemize}
Thus, in the second case, the branch for which $N$ (resp. $W$) is assumed, increases the ideal $\li N \ri\subset K[{\bf a}]$ (resp. $\li W \ri  \subset K[{\bf a}]$). Note that we replace $N$ by a Gr\"obner basis of its ideal (see line 4 in {\sc NewCond}). Since the ascending chains of ideals stabilize, the algorithm terminates. This argument was inspired by the proof in \cite{monts1}, Theorem 16.

To prove the {\em correctness}, assume that $M=\{(G_i,N_i,W_i)\}_{i=1}^{\ell}$ is the output of {\sc ComInvSys} for the input is $F$ (note that we have used the fact the this algorithm terminates in finitely many steps). Consider integer $1 \le i \le \ell$ homomorphism  $\sigma:K[{\bf a}]\rightarrow K^\prime$ where $(N_i, W_i)$ is a specification of $\sigma$ and  $K'$ is a field extension of $K$.

We have to show that for each $f\in G_i$ and each $x\in
NM_\L(\LM(\sigma(f)),\LM(\sigma(G_i)))$, in accordance with
Theorem \ref{blin}, the equality
$\nf_\L(\sigma(xf),\sigma(G_i))=0$ holds. By using `reductio ad
absurdum', suppose $g=\nf_\L(\sigma(xf),\sigma(G_i))$ and $g\ne
0$. Since $(G_i,N_i,W_i)$ has been added to {\tt List} in {\sc
Branch}, the leading coefficients of the polynomials in the
subalgorithm {\sc GBI}, examined at computation of a minimal
involutive basis for $F$, are decidable w.r.t. $(N_i, W_i)$.
Furthermore, $f\in G_i$ implies that in the course of GBI $xf$ is
added to $Q$, the set of all nonmultiplicative prolongations that
must be examined (see the notations used in {\sc GBI}). Then, {\sc
HeadReduce} is called to perform the $\L$-head reduction of the
elements of $Q$ modulo the last computed basis $T\subset G_i$. The
computed $\L$-head normal form of $xf$ is further reduced by
invoking {\sc TailNormalForm} which  performs the $\L$-tail
reduction. By the above notations, $g$ is the result of this step.
Thus, $g$ should be added to $T\subset G_i$. It follows that
$\nf_\L(\sigma(xf),\sigma(G_i))=0$, a contradiction, and this
completes the proof. $\Box$
\end{proof}

\section{Example}
\label{ex}

Now we give an example to illustrate the step by step construction of a minimal CIS by the algorithm {\sc ComInvSys} proposed and described in the previous section\,\footnote{The {\tt Maple} code of our implementation of the algorithm for the Janet division~\cite{gerdt0} is available at {\tt http://invo.jinr.ru} and {\tt http://amirhashemi.iut.ac.ir/software.html}}.

For the input $F=\{ax^2,by^2\}\subset \mathbb{K}[a,b,x,y]$ from Example \ref{ex2}, Janet division and the lexicographic monomial ordering with $b\prec_{\lex} a$ and $y\prec_{\lex}x$ the algorithm performs as follows:\\
\\
$\rightarrow ${\sc ComInvSys}$(F,\L,\prec_{lex},\prec_{lex})$\\[0.06cm]
$ \hspace*{0.4cm} {\rm {\tt List}}:=Null;\ {\rm {\tt ind}}:=1;\ k:=2;$\\[0.06cm]
$\hspace*{0.4cm} B:=\{[ax^2,ax^2,\emptyset],[by^2,by^2,\emptyset]\}$\\[0.06cm]
$\hspace*{0.4cm} \rightarrow ${\sc Branch}$([ax^2,ax^2,\emptyset],B,\{~\},\{~\},\{~\})$\\
$\hspace*{0.8cm}  \rightarrow ${\sc NewCond}$(ax^2,\{~\},\{~\})=(\{a\},\{~\},\{~\})$\\
$\hspace*{0.8cm} \rightarrow  ${\sc Branch}$([ax^2,ax^2,\emptyset],B,\{~\},\{a\},\{~\})$\\
$\hspace*{1.2cm}  \rightarrow ${\sc NewCond}$(ax^2,\{~\},\{a\})=(\{~\},\{~\},\{a\})$\\[0.06cm]
$ \hspace*{0.4cm} G:=\{([ax^2,ax^2,\emptyset],B,\{~\},\{a\},\{~\})\}$\\
$\hspace*{0.8cm} \rightarrow  ${\sc Branch}$([ax^2,ax^2,\emptyset],B,\{a\},\{~\},\{~\})$\\
$\hspace*{1.2cm}  \rightarrow ${\sc NewCond}$(ax^2,\{a\},\{~\})=(\{~\},\{a\},\{~\})$\\[0.06cm]
$ \hspace*{0.4cm} G:=\left\{{\big(}[ax^2,ax^2,\emptyset],B,\{~\},\{a\},\{~\}{\big)},{\big(}[ax^2,ax^2,\emptyset],\{[0,0,\emptyset],[by^2,by^2,\emptyset]\},\{a\},\{~\},\{~\}{\big)}\right\}$\\[0.06cm]
$ \hspace*{0.4cm} {\rm {\tt ind}}:=2;$\\[0.06cm]
$ \hspace*{0.4cm} A={\big(}[ax^2,ax^2,\emptyset],B,\{~\},\{a\},\{~\}{\big)}$\\[0.06cm]
$\hspace*{0.4cm} \rightarrow ${\sc Branch}$([by^2,by^2,\emptyset],B,\{~\},\{a\},\{~\})$\\
$\hspace*{0.8cm}  \rightarrow ${\sc NewCond}$(by^2,\{~\},\{a\})=(\{b\},\{~\},\{~\})$\\
$\hspace*{0.8cm} \rightarrow  ${\sc Branch}$([by^2,by^2,\emptyset],B,\{~\},\{a,b\},\{~\})$\\[0.06cm]
\hspace*{1.0cm}  (* further {\sc Branch}$([by^2,by^2,\emptyset],B,\{b\},\{a\},\{~\})$ is executed*)\\[0.06cm]
$\hspace*{1.2cm}  \rightarrow ${\sc NewCond}$(by^2,\{~\},\{a,b\})=(\{~\},\{~\},\{a,b\})$\\[0.06cm]
$\hspace*{1.2cm}{\rm {\tt ind}}\geq k=2$\\
$\hspace*{1.2cm} cd=\{~\}$\\[0.06cm]
$\hspace*{1.2cm}  \rightarrow ${\sc GBI} $(B,\{~\},\{a,b\},\{~\})$\\[0.06cm]
$\hspace*{1.6cm} T:=\{[by^2,by^2,\emptyset]\}$\\
$\hspace*{1.6cm} Q:=\{[ax^2,ax^2,\emptyset]\}$\\[0.06cm]
$\hspace*{1.6cm}  \rightarrow ${\sc HeadReduce}$(T,\{~\},\{a,b\},Q)$\\
$\hspace*{2cm}  \rightarrow ${\sc HeadNormalForm}$([ax^2,ax^2,\emptyset],T,\{~\},\{a,b\})=(true,ax^2,T,\{~\},\{a,b\})$\\[0.06cm]
 $\hspace*{1.6cm}${\sc HeadReduce} returns $(true,0,T,\{~\},\{a,b\},Q)$\\[0.06cm]
$\hspace*{1.6cm} p:=[ax^2,ax^2,\emptyset]$\\
$\hspace*{1.6cm} Q=\{~\}$\\[0.06cm]
$\hspace*{2cm}  \rightarrow ${\sc TailNormalForm}$(p,T)=ax^2$\\[0.06cm]
$\hspace*{1.6cm} T:=\{[by^2,by^2,\emptyset],[ax^2,ax^2,\emptyset]\}$\\
$\hspace*{1.6cm} Q:=\{[bxy^2,by^2,\emptyset]\}$\\[0.06cm]
$\hspace*{1.6cm}  \rightarrow ${\sc HeadReduce}$(T,\{~\},\{a,b\},Q)=(true,0,T,\{~\},\{a,b\},Q)$\\[0.06cm]
$\hspace*{1.6cm} p:=[bxy^2,by^2,\emptyset]$\\
$\hspace*{1.6cm} Q=\{~\}$\\[0.06cm]
$\hspace*{2cm}  \rightarrow ${\sc TailNormalForm}$(p,T)=bxy^2$\\[0.06cm]
$\hspace*{1.6cm} T:=\{[by^2,by^2,\emptyset],[ax^2,ax^2,\emptyset],[bxy^2,by^2,\emptyset]\}$\\
$\hspace*{1.6cm} Q:=\{[bx^2y^2,by^2,\emptyset]\}$\\[0.06cm]
$\hspace*{1.6cm}  \rightarrow ${\sc HeadReduce}$(T,\{~\},\{a,b\},Q)=(true,0,T,\{~\},\{a,b\},\{~\})$\\[0.06cm]
$\hspace*{1.6cm} Q:=\{~\}$\\[0.06cm]
$\hspace*{1.2cm}  \rightarrow ${\sc GBI}  returns $(true,0,\{by^2,ax^2,bxy^2\},\{~\},\{a,b\})$\\[0.06cm]
$\hspace*{1.2cm}  {\rm {\tt List}}:=(\{by^2,ax^2,bxy^2\},\{~\},\{a,b\})$\\[0.06cm]
$ \hspace*{0.4cm} B=\{[ax^2,ax^2,\emptyset],[0,0,\emptyset]\}$\\[0.06cm]
$\hspace*{0.8cm} \rightarrow  ${\sc Branch}$([by^2,by^2,\emptyset],B,\{b\},\{a\},\{~\})$\\
$\hspace*{1.2cm}  \rightarrow ${\sc NewCond}$(by^2,\{b\},\{a\})=(\{~\},\{b\},\{a\})$\\[0.06cm]
$\hspace*{1.2cm}{\rm {\tt ind}}\geq k=2$\\
$\hspace*{1.2cm} cd=\{~\}$\\[0.06cm]
$\hspace*{1.2cm}  \rightarrow ${\sc GBI} $(B,\{b\},\{a\},\{~\})=(true,0,\{ax^2\},\{b\},\{a\})$\\[0.06cm]
$\hspace*{1.2cm}  {\rm {\tt List}}:=(\{by^2,ax^2,bxy^2\},\{~\},\{a,b\}),(\{ax^2\},\{b\},\{a\})$\\[0.06cm]
\hspace*{1.2cm} (* Return back to {\sc ComInvSys} *)\\[0.06cm]
$ \hspace*{0.4cm} A={\big(}[ax^2,ax^2,\emptyset],\{[0,0,\emptyset],[by^2,by^2,\emptyset]\},\{a\},\{~\},\{~\}{\big)}$\\[0.06cm]
$ \hspace*{0.4cm} B=\{[0,0,\emptyset],[by^2,by^2,\emptyset]\}$\\[0.06cm]
$\hspace*{0.4cm} \rightarrow ${\sc Branch}$([by^2,by^2,\emptyset],B,\{~\},\{a\},\{~\})$\\
$\hspace*{0.8cm}  \rightarrow ${\sc NewCond}$(by^2,\{a\},\{~\})=(\{b\},\{~\},\{~\})$\\
$\hspace*{0.8cm} \rightarrow  ${\sc Branch}$([by^2,by^2,\emptyset],B,\{a\},\{b\},\{~\})$\\
\hspace*{0.8cm} (* further {\sc Branch}$([by^2,by^2,\emptyset],B,\{a,b\},\{~\},\{~\})$ is executed *)\\[0.06cm]
$\hspace*{1.2cm}  \rightarrow ${\sc NewCond}$(by^2,\{a\},\{b\})=(\{~\},\{a\},\{b\})$\\[0.06cm]
$\hspace*{1.2cm}{\rm {\tt ind}}\geq k=2$\\
$\hspace*{1.2cm} cd=\{~\}$\\[0.06cm]
$\hspace*{1.2cm}  \rightarrow ${\sc GBI} $(B,\{a\},\{b\},\{~\})=(true,0,\{by^2\},\{a\},\{b\})$\\[0.06cm]
$\hspace*{1.2cm}  {\rm {\tt List}}:=(\{by^2,ax^2,bxy^2\},\{~\},\{a,b\}),(\{ax^2\},\{b\},\{a\}),(\{by^2\},\{a\},\{b\})$\\[0.06cm]
$ \hspace*{0.4cm} B=\{[0,0,\emptyset],[0,0,\emptyset]\}$\\
$\hspace*{0.8cm} \rightarrow  ${\sc Branch}$([by^2,by^2,\emptyset],B,\{a,b\},\{~\},\{~\})$\\
$\hspace*{1.2cm}  \rightarrow ${\sc NewCond}$(by^2,\{a,b\},\{~\})=(\{~\},\{a,b\},\{~\})$\\[0.06cm]
$\hspace*{1.2cm}{\rm {\tt ind}}\geq k=2$\\
$\hspace*{1.2cm} cd=\{~\}$\\[0.06cm]
$\hspace*{1.2cm}  \rightarrow ${\sc GBI} $(B,\{a,b\},\{~\},\{~\})=(true,0,\{0\},\{a,b\},\{~\})$\\[0.06cm]
${\rm {\tt List}}:=(\{by^2,ax^2,bxy^2\},\{~\},\{a,b\}),(\{ax^2\},\{b\},\{a\}),(\{by^2\},\{a\},\{b\}),(\{0\},\{a,b\},\{~\})$.\\

\section*{Acknowledgements}
The main part of research presented in the paper was done during
the stay of the second author (A.H.) at the Joint Institute for
Nuclear Research in Dubna, Russia. He would like to thank the
first author (V.G.) for the invitation, hospitality, and support.
The contribution of the first author was partially supported by
grants 01-01-00200, 12-07-00294 from the Russian Foundation for
Basic Research and by grant 3802.2012.2 from the Ministry of
Education and Science of the Russian Federation.

\bibliographystyle{alpha}

\end{document}